\documentclass[reqno]{amsart}
\usepackage{amsfonts,amsmath,amssymb}
\usepackage{color}
\usepackage[ruled]{algorithm2e}
\usepackage{graphicx}
\usepackage{textcomp}
\usepackage{adjustbox}
\usepackage{url}
\usepackage{hyperref}  

\newtheorem{thm}{Theorem}[section]

\newtheorem{lem}{Lemma}[section]

\newtheorem{rem}{Remark}[section]

\newcommand{\R}{\mathbb{R}}
\newcommand{\E}{\mathbb{E}}
\newcommand{\N}{\mathbb{N}}
\newcommand{\ve}{\varepsilon}

\begin{document}

\title[A game theoretic model of wealth distribution]{A game theoretic model of wealth distribution}
\author[J. P. Pinasco, M. Rodriguez Cartabia, N. Saintier]
{Juan Pablo Pinasco, Mauro Rodriguez Cartabia, Nicolas Saintier}%

\begin{abstract}
In this work we   consider an agent based model in order to study the wealth distribution
problem where the interchange is determined with a symmetric zero sum game. Simultaneously, the
agents update their way of play trying to learn the optimal  one. Here, the agents use mixed
strategies. We study this model using both simulations and theoretical tools. We derive the equations for the learning mechanism,  and we show  that
the mean strategy of the population satisfies an equation close to the classical replicator equation.

Concerning the wealth distribution, there are two interesting situations depending on the equilibrium of the game. If the equilibrium is a pure
strategy, the wealth distribution is fixed after some transient time, and those players which are close to optimal strategy are richer. When the game
has an equilibrium in mixed strategies, the  stationary wealth distribution is close to a Gamma distribution with
variance  depending on the coefficients of the game matrix. We compute theoretically their
second moment in this case. \end{abstract}

\maketitle


\section{Introduction}

The stability through time and societies of the wealth distribution in a population has
attracted a lot of attention since Pareto's seminal work \cite{Pareto}. Different empirical
studies show that in many economies the distribution of wealth follows a Pareto distribution
(a power law) in the high-income range, and a Gibbs, a Gamma or a log-normal distribution in
the bulk-range. We refer the reader to Chapter 2 in \cite{CCCC} for a review of the empirical
data available.

These empirical results has given rise to theoretical efforts trying to explain this remarkable
universality of the wealth distribution. In particular, recent advances in computer technology
allowed  the study of this phenomenon computationally by using agent-based models. The main
idea consists in consider that the distribution of wealth across the population results from
simple repeated interactions between individuals. This is reminiscent of the kinetic theory of
gases where the particles composing some gas interact binary at random exchanging energy
at each interaction, and, as a result of the interactions, the systems tends to some stationary
state as time increases. The analogy consists then in thinking of a population as composed
of many agents, and when two agents interacts, they exchange money following some
pre-established rule. The wealth distribution is expected to approach some stationary state,
 which is expected to look similar as the wealth distribution empirically
observed if the interaction rule is appropriately chosen. There are many
numerical and theoretical studies which shows that this kind of models can reproduce the
 curves observed in wealth distribution, in particular the Pareto tail and the exponential
or Gamma-like bulk-distribution, see for instance  Chapter 4 in \cite{CCCC} for a detailed review of
different models.

\bigskip

Two main criticism of these models inspired by and analyzed with statistical mechanics
arguments comes from the fact that agents are considered as identical dummy particles and
hence they are not acting rationally trying to improve their wealth.  Also, they interact always
in the same way and thus they cannot modify or adapt their behaviour. This is in a sharp
contrast with economical models, where  game theory provides a powerful framework to
study interactions between rational agents through games in which the players can choose
among different actions or strategies.

\bigskip
In this article we propose  a model of wealth distribution taking advantage of both the
simplicity of the statistical physics models,
 and the
flexibility of evolutionary game theory. This approach to conflict resolution in the animal
kingdom was proposed in the late '70s by John Maynard Smith, and it enables the players to
modify its behavior trying to adapt their way of play against the strategies of others players,
see for instance \cite{F,H,S} for comprehensive introductions in this subject.

In our model, agents interact pairwise, playing some zero sum game, and the outcome of the
game defines the money exchange between them. More general games requires some
renormalization, since money is not conserved, as in \cite{PaTo}. Also, the agents update
their strategies based on this outcome, following some adaptive mechanism that we
introduce below.

We assume  that each player has a mixed strategy, and they will increment or decrease
 the probability to play some strategies according to the result of the
last game they played. This is a Pavlovian type of agent, introduced by D. and V. Kraines in
\cite{Kraines}, which has a limited forecast of the game, and react myopically. Despite that
this kind of behavior has been used in many works, specially through simulations, see
\cite{Bou, Silva, Szi}, a theoretical analysis is lack, and we provide here the corresponding
dynamics whenever finitely many agents are involved. This strategic update has an
independent interest, and we study here their main properties. In particular, we show that
the mean strategy of the population obeys an equation close to the replicator equation
introduced by Taylor and Jonker \cite{TJ} (see also \cite{Hofb, Ze}).

\bigskip
We present both theoretical and computational arguments which confirms   that this model
reproduce the main features of the observed wealth distribution. Also, by scaling the game
payoff we obtain an analogue of the {\it saving propensity}: let us recall that in \cite{vercual},
a family of parameters $\{\lambda_i\}_{1\le i\le N}$ is introduced, and player $i$ is willing to
exchange a fraction $\lambda_i w_i$ of  its wealth, without exposing  the remaining
$(1-\lambda_i) w_i$, see  the review \cite{Patriarca} for a detailed analysis of the influence
of saving propensities on the shape of the equilibrium wealth distribution. As we will show in
the paper, the payoff scaling influence the wealth distribution in similar ways, by changing
the mean and variance of the equilibrium distributions.

Finally, let us remark that previous models considered some gambling process associated to
the money exchange process,
 see \cite{167, 259, 270}. However, the outcome of the random game was externally determined, and agent's expectations
 were always the same, since they cannot improve their chances in the game.

\bigskip

Our paper is organized as follows. In Section \S 2 we describe the model, and in Section \S 3
we present the theoretical results on the evolutionary mechanism.  In Section \S 4 we present
some examples of this dynamic. Section \S 5 is devoted to the wealth  distribution, and we
compare the theoretical predictions with our computer simulations. We close the paper with
some conclusions and comments in Section \S 6.

\section{Description of the model}

Let us consider a toy model of a closed economic system where the total amount of
money and the total number $N$ of agents are fixed. No changes occur in those variables: no production, neither migration,
nor death or birth of agents,  and the only economic activity is restricted to trading.
We denote the wealth of agent  $i$ by $w_i$, which is always a nonnegative real number.

\subsection{Mechanism of a trade}

Let us first describe the trade mechanism. Trades   take place between two
agents. They play a  game whose payoff matrix $G$ will indicate how wealth is transferred
between them. We consider only  zero-sum games, and therefore the gain of one agent is the
loss of the other. This will ensure that the wealth is conserved. The winner will receive some
fraction of the wealth of the other, this
 fraction depends on the
game and the strategies selected by the players.

\subsubsection{The game}

We assume we have a symmetric zero sum game. Each agent has a common set of $K$ possible actions (or
pure strategies) labeled $s_1, \ldots, s_K$. The
game matrix $G = (g_{lm}) \in \R^{K\times K}$ has entries $g_{lm}\in[-1,1]$, and $g_{lm}=-g_{ml}$ for
any $1\le l,m\le K$. Let us note that this is not a serious restriction, since any zero sum game
can be extended to a symmetric game, where each agent is randomly assigned as the row or column player.

Agent $i$ has a private probability distribution $p^{(i)}=(p^{(i)}_{1},\ldots,p^{(i)}_{K})$ on
the actions, where $p^{(i)}_{l}$ is the probability that agent $i$ chooses the pure strategy
$s_l$, with $1\le l \le K$. We have $0\le p^{(i)}_{l}\le 1$ for any $1\le l \le K$ and
$p^{(i)}_{1}+...+p^{(i)}_{K}=1$ for any agent $i$.

In each time step, two agents $i$ and $j$ are randomly selected with uniform probability, and
they play some strategies  selected with their respective private probabilities distributions.

\subsubsection{Wealth interchange}

 Let us assume that $i$ plays $s_l$, and $j$ plays $s_m$. Then, $i$ will receive the fraction
$g_{lm}$ of the wealth $w_j$ of $j$ if $g_{lm}\ge 0$; otherwise $i$ will give to $j$
the fraction $g_{ml}$ of its own wealth $w_i$.

Hence, the post-interaction wealth $w_i^*$ and $w_j^*$ are defined as follows:
\begin{equation}\label{UpDatePayOff}
 \begin{cases} w_i^* = w_i +   \Delta    \\
                 w_j^*=w_j -   \Delta
		\end{cases}
 \quad \textit{where } \Delta := g_{lm}^+w_j - g_{lm}^-w_i.
\end{equation}
Here, $x^+:=\max\{x,0\}$ and $x^-:=\max\{-x,0\}$ .

Observe that $g_{lm}\in[-1,1]$ implies that the wealth of each agent remains
non-negative, while  condition $G^T=-G$ implies that the trade mechanism
is symmetric in the sense that the role of the players is unimportant. In particular, the
diagonal elements of $G$ are $0$, and if both players choose the same strategy  then there
are no money transfer.

\bigskip
A more realistic model must introduce some degree of risk aversion. An usual approach  \cite{vercual, Patriarca} is to assign
a {\it saving propensity} $\lambda_i \in [0,1]$ to each player, and in the trade   player $i$ is willing to exchange only
the fraction $\lambda_i w_i$ of  its wealth. Here, we consider this situation by  re-scaling  the game using a parameter $\ve\in [0,1]$, and
we have
$$
 \begin{cases} w_i^* = w_i +   \Delta_\ve   \\
                 w_j^*=w_j -   \Delta_\ve
		\end{cases}
 \quad \textit{where } \Delta_\ve:=\ve(g_{lm}^+w_j - g_{lm}^-w_i).
$$
This have two possible interpretations: it is equivalent to change the game matrix to $\ve G$, obtaining a less dangerous game since
lower fractions of wealth are at risk; or each player is exposing the fraction
$\ve w$ of its wealth. Of course, $\ve$ can be randomly selected from some distribution in each trade; and also each player
can have its own value $\ve$.

\subsection{Adaptive process}

  We assign to each player an initial probability vector $p^{(i)}=(p_1^{(i)},\cdots, p_K^{(i)})$, which correspond to its mixed strategy,
   and let us fix some small value $\delta>0$.

We assume that all the agents will update  their mixed strategies  $\{p^{(i)}\}_{1\le i \le N}$ trying to increase the pay-off in
future trades. Now, the Pavlovian character of the agents means that  according to the outcome of the
game they will slightly increase (respectively, decrease) by $\delta$ the probability of
playing the successful (resp., unsuccessful) strategy used in that game.

More precisely, we have the following updating rules:
 {\it whenever agents $i$ and $j$ play strategies $s_l$ and $s_m$ respectively, and $\Delta$ is given in equation
 \eqref{UpDatePayOff},
  }
\begin{equation}\label{UpDateProba1}
\begin{array}{l}
\textit{if $\Delta >0$ then }  \begin{cases}
  \delta^{(i)} = \min\{\delta, p^{(i)}_m \}  \\ 
 \delta^{(j)} = \min\{\delta,  p^{(j)}_m \}  \\
  p^{(i)*}_{l}= p^{(i)}_{l}+\delta^{(i)},  \qquad p^{(i)*}_{m}  = p^{(i)}_{m} -\delta^{(i)}  \\
p^{(j)*}_{l} =                     p^{(j)}_{l}+\delta^{(j)},\qquad  p^{(j)*}_{m} = p^{(j)}_{m}-\delta^{(j)}, \\
                        		\end{cases}
\\
\\
\textit{if $\Delta <0$ then }
 \begin{cases}
  \delta^{(i)} = \min\{\delta, p^{(i)}_l \}  \\
 \delta^{(j)} = \min\{\delta, p^{(j)}_l\}  \\  p^{(i)*}_{l}= p^{(i)}_{l}-\delta^{(i)},  \qquad  \, p^{(i)*}_{m} = p^{(i)}_{m} +\delta^{(i)}  \\
p^{(j)*}_{l} =    p^{(j)}_{l}-\delta^{(j)},\qquad \,  p^{(j)*}_{m}  = p^{(j)}_{m}+\delta^{(j)}.\\
		\end{cases}
\end{array}
\end{equation}

 Clearly, $\delta^{(i)}=\delta$ unless the winning (resp., loosing) strategy has probability $1$ (resp., $0$), because in these
 cases the update procedure fails to give a true probability vector.

\medskip

In order to simplify the simulations, we can fix some positive  number $M\in \N$, and
now we choose  $\delta=1/M$. If we assign to each player an initial probability vector $p^{(i)}$  such that their
components are integer multiples of $\delta$, the update process is simpler, since $\delta=\delta^{(i)}=\delta^{(j)}$ unless the
player strategy has reached the values $0$ or $1$.

\subsection{The algorithm}

The pseudo code in Algorithm 1 can be easily implemented in any programming language.
We have implemented it and run our simulations in GNU Octave \cite{Oct}.

Let us note that Step 1 can be replaced by choosing one of the agents sequentially from $1$
to $N$, and selecting only the second agent at random. The results in both cases  are strongly similar.

\begin{algorithm}
\KwData{\\ $K\in\N$: number of actions or pure strategies, \\
$G\in\R^{K\times K}$: antisymmetric matrix with coefficients in $[-1,1]$.\\
$N\in\N$: number of agents. \\
$T\in\N$: total number of trades. \\
$\delta>0$ small (to be used in \eqref{UpDateProba1}) \\   } {\bf Initialization:} for each
agent $i=1,\ldots,N$, initialize its wealth $w_i\ge 0$ and its mixed strategy $p^{(i)}$. \\ \For{$t=1,\ldots,T$} { 1. Select uniformly at
random two agents $i$ and $j$.  \\ 2. For agent $i$, choose at random a pure strategy
following its probabilities $p^{(i)}_{1},\ldots ,p^{(i)}_{K}$.   \\ 3. Do the same for agent $j$.
\\4. Update their wealth $w_i$ and $w_j$ following \eqref{UpDatePayOff}.  \\ 5. Update their
probabilities following \eqref{UpDateProba1}.   } \KwResult{ final wealth $w_i$ and
mixed strategy $p^{(i)}$  of each agent $i\le i\le N$. } \caption{pseudo code
\label{BasicModel}}
\end{algorithm}

\section{Evolutionary game theory}

Let us  fix some generic player $i$, and let us study the evolution of the mixed strategy
$p^{(i)}(t)$.

For notational convenience, let us introduce the  mean strategy of the population,
$$ \bar p = N^{-1} \sum_{i} p^{(i)}.$$
Also, let us introduce a matrix $H=(h_{lm})_{1\le l,m\le K}$, depending on $G$ as follows:
$$
h_{lm}=\left\{\begin{array}{rcl}
1 & \; & \mbox{ if } g_{lm}>0 \\
0 & \; & \mbox{ if } g_{lm}=0 \\
-1 & \; &  \mbox{ if } g_{lm}<0 \\
\end{array}\right.
$$
Let us note that only the sign of $\Delta$  matters in the update rule given in equation
\eqref{UpDateProba1}. Incidentally, it will have important consequences on the convergence
of the mean strategy.

Let us assume that a trade occurs following a Poisson distribution with parameter $\lambda=1$.
So, the probability that an interaction occurs in the time interval $(t,t+dt)$ is $1-e^{-dt}$. We  compute now the expected  change of the
probability $p_l^{(i)}$ that player $i$ assigns to  action $s_l$ in its mixed strategy.

Observe that player $i$ is selected with
probability $2/N$ (it is the first or the second player), and it plays the pure strategy $s_n$ with
probability $p^{(i)}_n(t)$. Another arbitrary player $j$ is selected at random with probability $(N-1)^{-1}$,
and it will play the pure strategy $s_m$ with probability $p^{(j)}_m(t)$.
After the game, the probability $p^{(i)}_l(t)$ is changed by $\delta^{(i)}h_{lm}$ when $i$ plays strategy $l$, or
by $-\delta^{(i)}h_{nl}$ if $i$ plays a different strategy and the rival plays $s_l$.

Recall that $\delta^{(i)}$ is allowed to be zero to prevent a value outside $[0,1]$, depending on
the values of and $p^{(i)}_l(t)$ and $p^{(i)}_m(t)$;
we will omit the supraindex for notational simplicity.

We get the following master equation, where $e_l^T$ denotes the transposed $l$-th
canonical vector:
\begin{align*}
p^{(i)}_l(t+dt)= & p^{(i)}_l(t) + 2 dt \delta  N^{-1}\Big[  \sum_{m\neq l}  p^{(i)}_l(t) h_{lm}
   \Big( \frac{1}{N-1}\sum_{j\neq i}p^{(j)}_m(t)  \Big) \Big] \\
  &  - 2 dt \delta N^{-1}\Big[  \sum_{n\neq l}  p^{(i)}_n(t) h_{nl}   \Big(\frac{1}{N-1}\sum_{j\neq i}p^{(j)}_l(t)  \Big) \Big]
\\=& p^{(i)}_l(t) + 2dt \delta  N^{-1} p^{(i)}_l(t) e_l^T H \Big(\frac{1}{N-1}\sum_{j\neq i}p^{(j)}(t)  \Big)\\
  &  - 2 dt \delta N^{-1}   p^{(i)} (t) H e_l \Big(\frac{1}{N-1}\sum_{j\neq i}p^{(j)}_l(t)  \Big),
\end{align*}
where we have used that $h_{ll}=0$ in the last step.

Therefore, when  $dt \to 0^+$,
and rescaling time in order to get rid off the term $2 \delta/(N-1)$, we can rewrite it as

$$\frac{d p^{(i)}_l }{dt}=
p^{(i)}_l  e_l^T H \Big(   \bar p  - \frac{1}{N}  p^{(i)} \Big)  - p^{(i)} H e_l\Big(  \bar p_l  - \frac{1}{N}p^{(i)}_l\Big).
$$

Moreover,   $p^{(i)}_l e_l^T H p^{(i)} =  \sum_m  p^{(i)}_l h_{lm}p_m^{(i)}
 = -\sum_m p^{(i)}_l h_{ml} p_m^{(i)} = -p^{(i)} H e_l p^{(i)}_l$
 since $H$ is antisymmetric, and we have,
for  $1\le l\le K$,
 \begin{equation}\label{nordinaria}
\frac{d p^{(i)}_l }{dt}=
p^{(i)}_l  e_l^T H   \bar p  -  p^{(i)} H e_l  \bar p_l  - \frac{2}{N} p^{(i)}_l e_l^T H p^{(i)}.
\end{equation}

We have obtained a system of nonlinear coupled differential equations governing the
probabilities updates for the mixed strategy of each player. Let us mention that, when  $N\to \infty$,
the previous equations can be simplified to
 \begin{equation}\label{ordinaria}
\frac{d p^{(i)}_l}{dt}=
  p^{(i)}_l(t) e_l^T H   \bar p(t)
 -  p^{(i)}(t) H e_l  \bar p_l(t).
\end{equation}

 Despite the ugly aspect of   Eqns. \eqref{nordinaria} and  \eqref{ordinaria}, they can be solved
explicitly  in several interesting cases. As an example, we will study later two prototypical cases,  $2\times
2$ games, and the rock-paper-scissors game.

\bigskip
However, as we will see later in the analysis of the steady state of the wealth distribution problem, the
individual probabilities $p^{(i)}$ are unimportant,
 and the relevant information came
from their mean value $\bar p(t)$. So, we   need to derive a system of ordinary equations for $\bar p(t)$.  We have the following result:

\begin{thm} 
Let $\bar p$ be the   mean strategy of the population. Then   $\bar p$ satisfies the following equation
 \begin{equation}\label{repli}
\frac{d \bar p_l}{dt}  =    2 \bar p_l e_l^T H \bar p  -
2\bar p^T H\bar p  - \frac{2}{N^2} \sum_i p^{(i)}_l e_l^T H p^{(i)}.
\end{equation}
\end{thm}

\begin{proof}
For any $N$ fixed, we have
$\bar p = N^{-1}\sum_i p^{(i)}$ and
\begin{align*}
\frac{d \bar p_l}{dt} = & N^{-1} \sum_i \frac{d p^{(i)}_l}{dt} \\
=&  N^{-1} \sum_i
  \left[
p^{(i)}_l  e_l^T H   \bar p  -  p^{(i)} H e_l  \bar p_l  - \frac{2}{N} p^{(i)}_l e_l^T H p^{(i)} \right]\\
=&
\bar p_l  e_l^T H   \bar p  -  \bar p H e_l  \bar p_l  - \frac{2}{N^2} \sum_i p^{(i)}_l e_l^T H p^{(i)}.\\
\end{align*}

Now, using that $H^T=-H$, we have
\begin{align*}
\bar p_l  e_l^T H   \bar p = & \sum_m \bar p_l h_{lm} \bar p_m,\\
 \bar p H e_l  \bar p_l = & \sum_m \bar p_m h_{ml} \bar p_l=-\sum_m \bar p_m h_{lm} \bar p_l,
 \end{align*}
and we get
$$
\frac{d \bar p_l}{dt} =
2
\bar p_l  e_l^T H   \bar p - \frac{2}{N^2} \sum_i p^{(i)}_l e_l^T H p^{(i)} .
$$

We can rewrite the equation as
 $$
 \frac{d \bar p_l}{dt} =
2
\bar p_l  e_l^T H   \bar p - 2\bar p^T H \bar p - \frac{2}{N^2} \sum_i p^{(i)}_l e_l^T H p^{(i)} ,$$
using, as before, that $ \bar p^T H \bar p= \sum_{l}\sum_m
\bar p_l \bar p_m h_{lm}=0 $ since $h_{lm}=-h_{ml}$,
 and the proof is finished.
\end{proof}

\begin{rem}
Let us observe that, when $N\to \infty$, and after rescaling time to cancel the factor
$2$, we obtain that $\bar p(t)$ satisfy the {\it replicator equation},
\begin{equation}\label{replicator}
\frac{d \bar p_l}{dt} = \bar p_l  e_l^T H \bar p   -   \bar p^TH \bar p.
\end{equation}
As a consequence, the vast literature in this equation
enable us to characterize the
 convergence of $\bar p$
to a precise vector which correspond to a Nash equilibrium in several interesting cases.
 However, let us recall that there are fixed points of equation
\eqref{replicator} which are not Nash equilibria, and not every Nash equilibrium is an
asymptotically stable equilibrium,
see \cite{CresTao} for a short recent survey.
\end{rem}

\section{Two particular games and finite size effects}

Let us analyze the dynamics of the previous section for $2\times 2$ games and the family of generalized rock-paper-scissors games.

\subsection{$2\times 2$ games} \label{subdospordos}

\,

Let us suppose that we have a nontrivial game with only two strategies. There are
two possible matrices $H$ associated to the game, let us assume that
$$
H=\left(\begin{array}{rr} 0 & 1 \\ -1 & 0\end{array}\right),
$$
the other one is $-H$, and the results can be easily adapted for this case.

For large $N$, using Equation \eqref{ordinaria} we have
$$
\frac{dp_1^{(i)}}{dt} =
  p^{(i)}_1(t) (1- \bar p_1(t))
 +(1- p^{(i)}_1(t) )    \bar p_1(t),
$$
Let us note that the equation for $p^{(i)}_2$ is not necessary  since $p^{(i)}_2=1-p_1^{(i)}$.

This problem has only two equilibria: $p^{(i)}=e_1$ for every player $i$, or  $p^{(i)} =e_2$ for every player $ i$. Also,
a simple perturbation analysis shows that only $p^{(i)}=e_1$ is stable.

In fact, near the equilibria $\bar p= e_2$, if we set $p_1^{(i)} \approx \varepsilon_i$,
then
$$
\frac{dp_1^{(i)}}{dt}  \ge
  \varepsilon_i  (1- a)
 +(1- \varepsilon_i )   a   >0
$$
for  $a= N^{-1}\sum \varepsilon_i \approx 0$,
and then $p_1^{(i)}$ increases.

\subsubsection{Finitely many players, $N>3$}
If we consider the finite size effect of the population, Equation \eqref{nordinaria}
reads
\begin{equation}\label{dospordos}
\frac{dp_1^{(i)}}{dt} =
 p^{(i)}_1 (1- \bar p_1)
 +(1- p^{(i)}_1 )    \bar p_1 -\frac{2}{N}  p^{(i)}_1 (1- p^{(i)}_1 )
\end{equation}
 Again, the equilibria are the same as before, 
 although the stability analysis is slightly harder.

If we call $p=(p_1^{(1)}, \cdots, p_1^{(N)})$, we have 
$$
\frac{dp}{dt} = F(p),
$$
where $F:\R^N\to \R^N$, where each component of $F$ is given by the right hand side of Equation \eqref{dospordos}.
In order to prove that  the  equilibrium $p=0$ is unstable, we linearize this system at the equilibrium point
and we show that all the eigenvalues of $DF(0)$
have strictly positive real part.

   Let us observe
the diagonal coefficients of $DF(0) $ are given by
$$
DF(0)_{ii} = 1 - \frac{1}{N}
$$
while the remaining coefficients are $
DF(0)_{ij}=1/N$ for $i\neq j$.

Being a symmetric matrix, we know that the eigenvalues are real.  This matrix cannot
have a negative eigenvalue, and we can give a shorter theoretical argument using that
$$
DF(0) = \left(1-\frac{2}{N}\right) Id + \frac{1}{N}  \,  1s
$$
where $Id$ is the identity matrix, and $1s$ is a matrix of all ones.

\begin{lem}
The smaller eigenvalue of $DF(0)$ is greater than $1-2/N$.
\end{lem}

\begin{proof}  The smaller eigenvalue $\lambda_1$ of $DF(0)$ can be characterized as
$$
\lambda_1 = \inf_{\{v\in \R^N : \|v\|=1\}} v^TDF(0)v
$$
Let us note that $Id$ is strictly definite positive, and $1s$ is
semidefinite positive, so
\begin{align*}
\lambda_1 = & \inf_{\{v\in \R^N : \|v\|=1\}} v^TDF(0)v \\
\ge & \left(1-\frac{2}{N}\right) \inf_{\{v\in \R^N : \|v\|=1\}}  v^T Idv +   \frac{1}{N} \inf_{\{v\in \R^N : \|v\|=1\}} v^T 1s v \\
\ge & \left(1-\frac{2}{N}\right),
\end{align*}
and the Lemma is proved.
\end{proof}

\bigskip
A similar argument, linearizing at $p=1$, gives that this is an stable equilibrium.

Therefore,
the players learn the optimal strategy  of the game, $p_1^{(i)}=1$.

\subsubsection{Rock-paper-scissors}\label{rps}

  Let us consider this popular  game, which is a very important example in evolutionary game theory.
The generalized  rock-paper-scissors game is defined through the following matrix,
$$
G=\left(\begin{array}{rrr} 0 & -a & b\\ b& 0& -a\\ -a & b & 0\end{array}\right),
$$
with $a$, $b >0$. The only Nash equilibrium  is $(1/3, 1/3, 1/3)$, and the solutions of the replicator equation  converge, diverge or rotate depending on the relationship
between $a$ and $b$, see Chapter 7 in \cite{H}.

Here, for the different winner/looser payoffs,  the matrix $H$ is always
$$
H=\left(\begin{array}{rrr} 0 & -1 & 1\\ 1& 0& -1\\ -1 & 1 & 0\end{array}\right),
$$
and in this case we have 
\begin{align*}
\frac{dp_1^{(i)}}{dt} = & p_1^{(i)} e_1 H \bar p = p_1^{(i)}(1-\bar p_1-2\bar p_2)+\bar p_1 (1- p_1^{(i)}-2 p_2^{(i)}), \\
\frac{dp_2^{(i)}}{dt} = & p_2^{(i)} e_2 H \bar p= -p_2^{(i)}(1-2\bar p_1-\bar p_2)-\bar p_2(1-2  p_1^{(i)}-  p_2^{(i)}),
\end{align*}
and by rescaling time, the equation for $\bar p$ 
 is given by
\begin{align*}
\frac{d\bar p_1}{dt} = & \bar p_1 (1-\bar p_1-2\bar p_2), \\
\frac{d\bar p_2}{dt} = & -\bar p_2(1-2\bar p_1- \bar p_2).
\end{align*} 
 Linearizing the previous systems
at  the Nash equilibria $(1/3, 1/3, 1/3)$,  we find that $p^{(i)}$ and $\bar p$ describe periodic orbits. So, the players strategies and the
mean strategy of the population fail to converge.

\bigskip
However, an striking point observed in the simulations was the convergence
$$ \bar p(t) \to (1/3, 1/3, 1/3)$$
as $t\to \infty$, see Figure 1. 

This seems to be a finite size effect, which can be understood by analyzing the system of equations given by \eqref{nordinaria}.
In this case, we get that $DF(0)\in \R^{2N\times 2N}$ is a block matrix, and computations with Octave, for different values of $N$, show
that it has purely imaginary eigenvalues.
Hence, the evolution of the mixed strategy of each agent follows a periodic orbit centered at $(1/3, 1/3, 1/3)$. However,
the random
perturbations of the orbits and the different initial positions generate a diffusive effect, and they are scattered around  the fixed point.
See the snapshots in Figure 2, where the  solid line correspond to the strategy of the mean strategy, which converges to $(1/3, 1/3, 1/3)$.


  \begin{figure}\label{fig1conver}
 \centering
   \includegraphics[angle=0,width=1\textwidth]{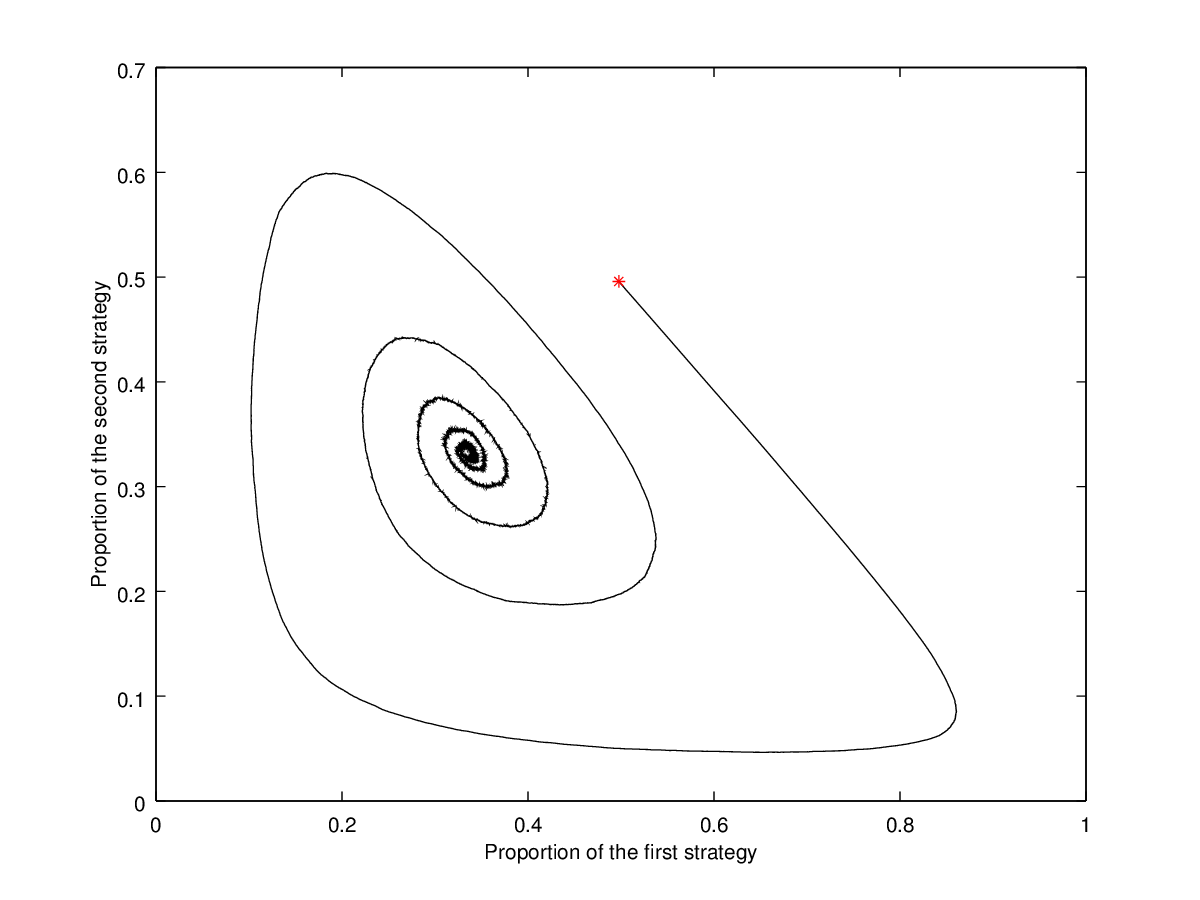}
 \caption{Evolution of $(\bar p_1, \bar p_2)$ for the rock-paper-scissors game. Initially,   $10^4$ players started near
 $(1/2,1/2,0)$.   In the simulation they performed $4 \times 10^6$ random matches with $\delta=0.05$.}
\end{figure}

\begin{figure}\label{fig2capturas}
\begin{center}
  \begin{minipage}{0.45\textwidth}
  \adjustbox{trim={.0\width} {0.2\height} {0.0\width} {0.25\height},clip}
    {\includegraphics[width=\textwidth]{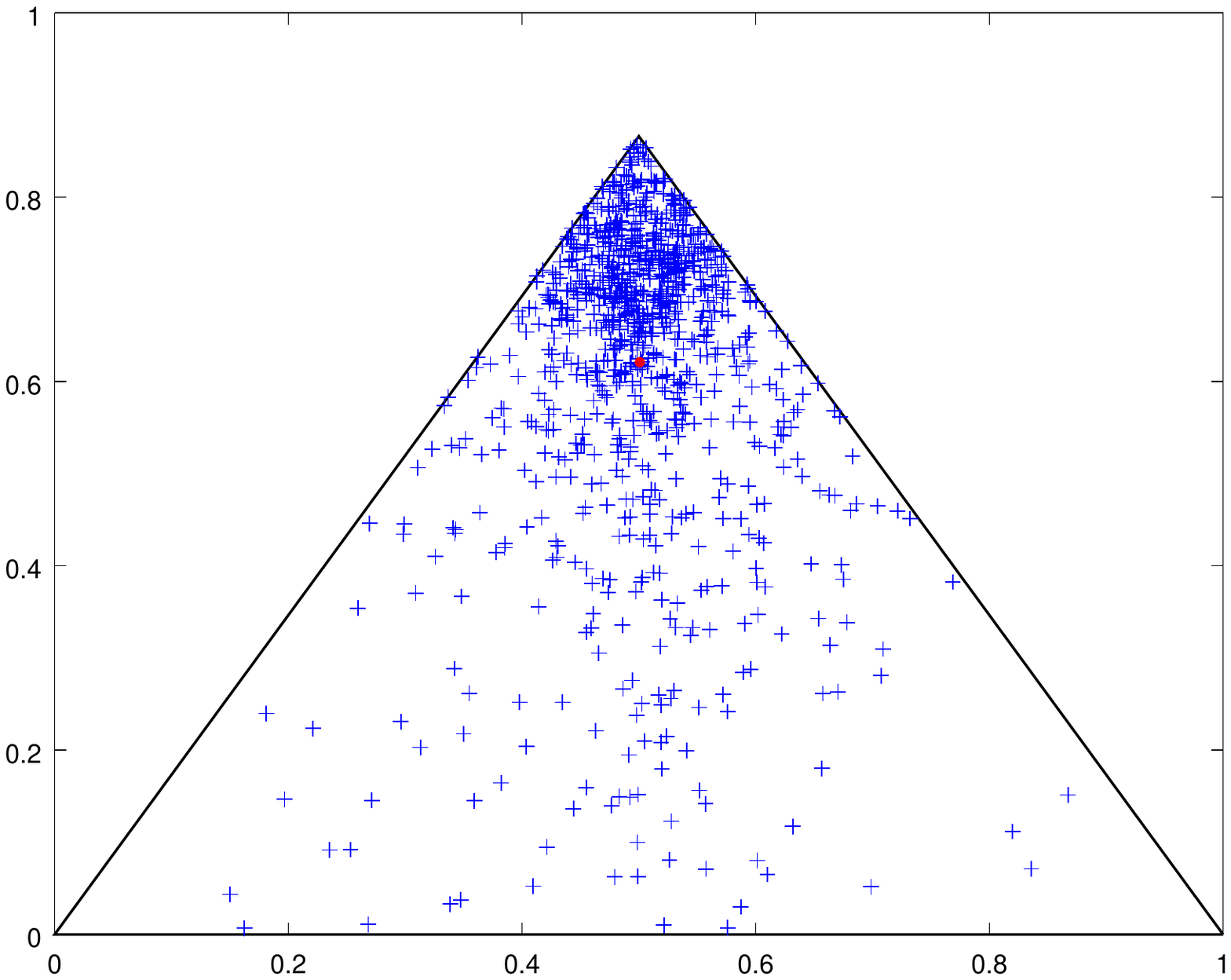}}
     \adjustbox{trim={.0\width} {0.25\height} {0.0\width} {0.2\height},clip}
    {\includegraphics[width=\textwidth]{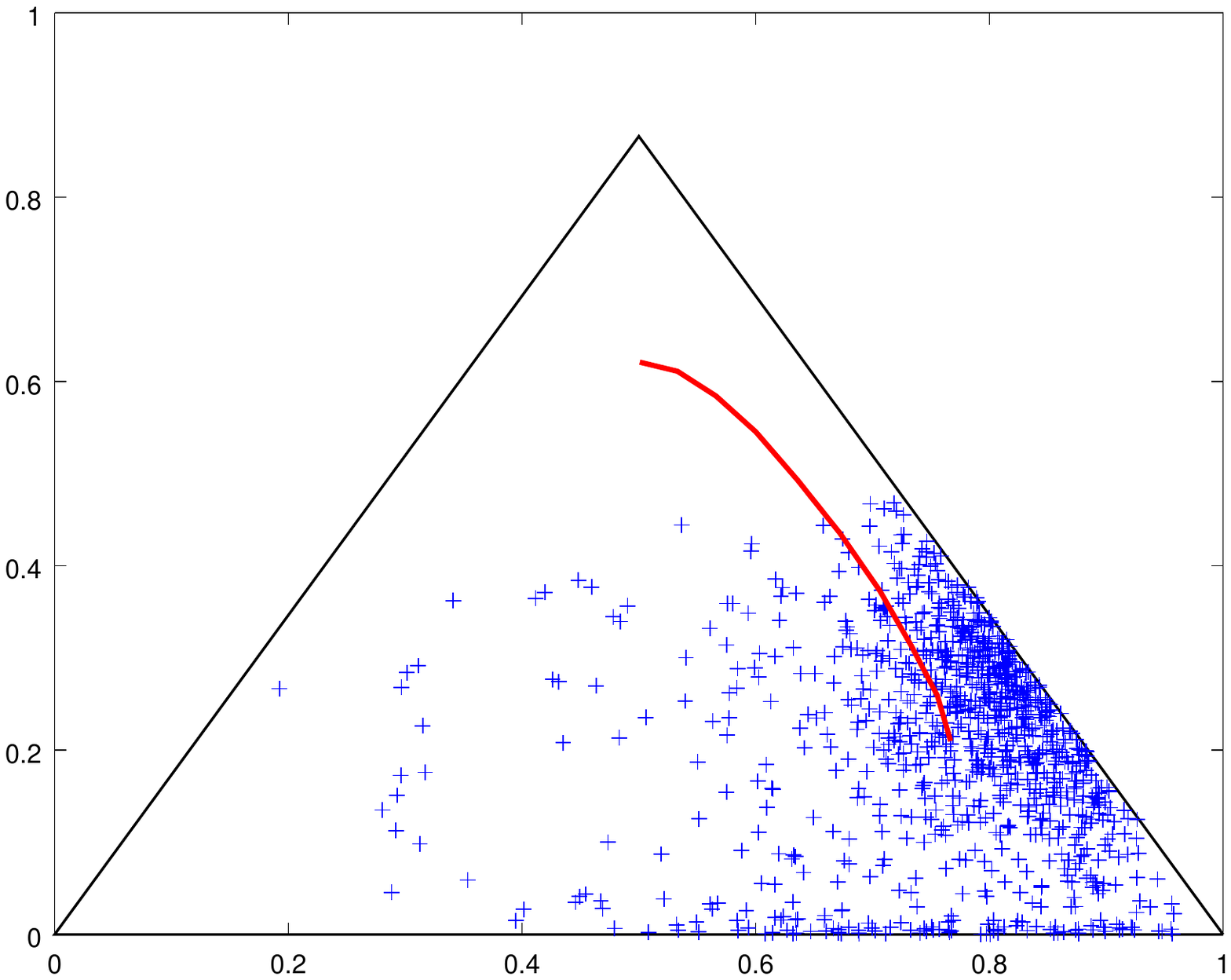}}
     \adjustbox{trim={.0\width} {0.25\height} {0.0\width} {0.2\height},clip}
    {\includegraphics[width=\textwidth]{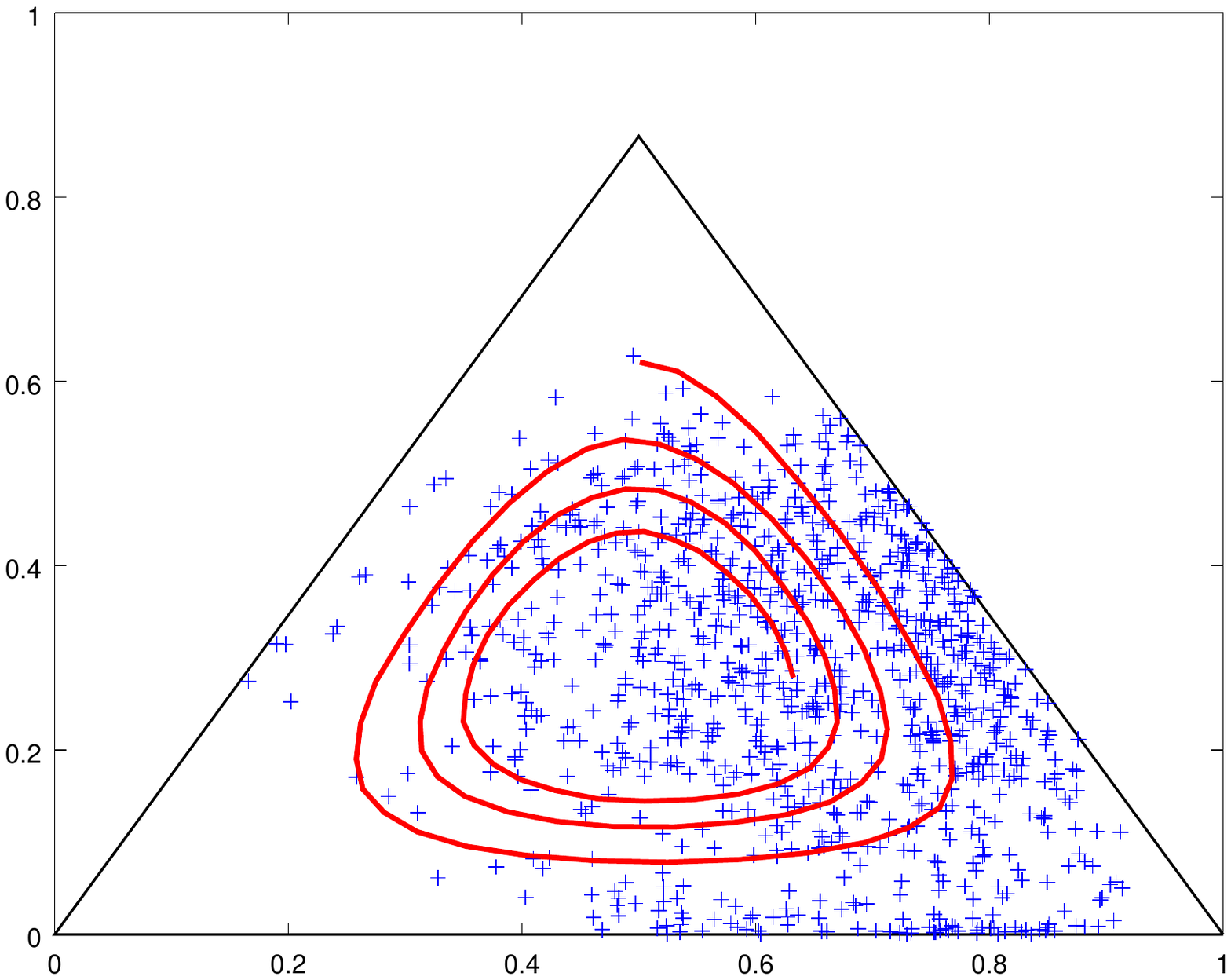}}
  \end{minipage}
  \hfill
  \begin{minipage}{0.45\textwidth}
   \adjustbox{trim={.0\width} {0.2\height} {0.0\width} {0.25\height},clip}
    {\includegraphics[width=\textwidth]{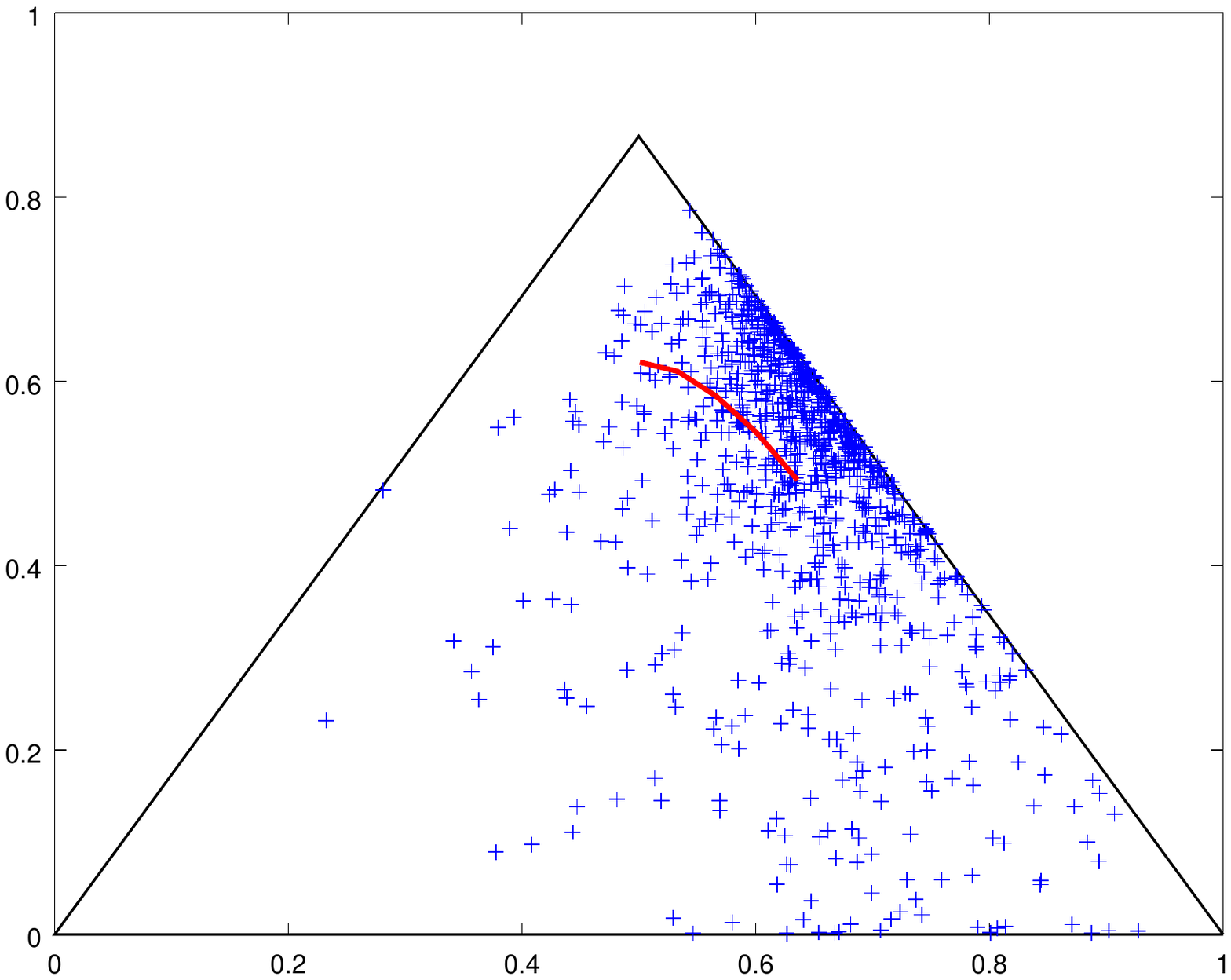}}
     \adjustbox{trim={.0\width} {0.25\height} {0.0\width} {0.2\height},clip}
    {\includegraphics[width=\textwidth]{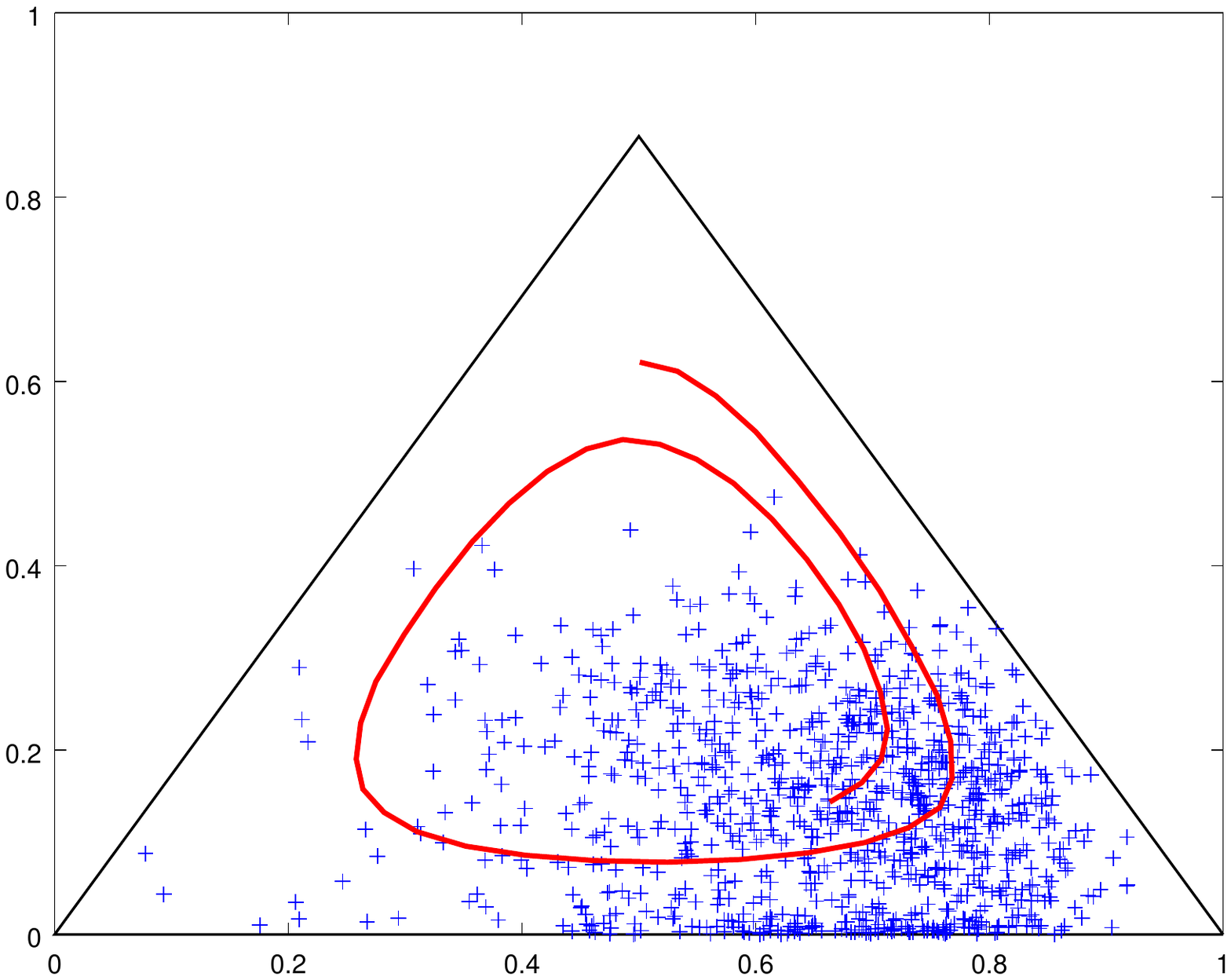}}
     \adjustbox{trim={.0\width} {0.25\height} {0.0\width} {0.2\height},clip}
    {\includegraphics[width=\textwidth]{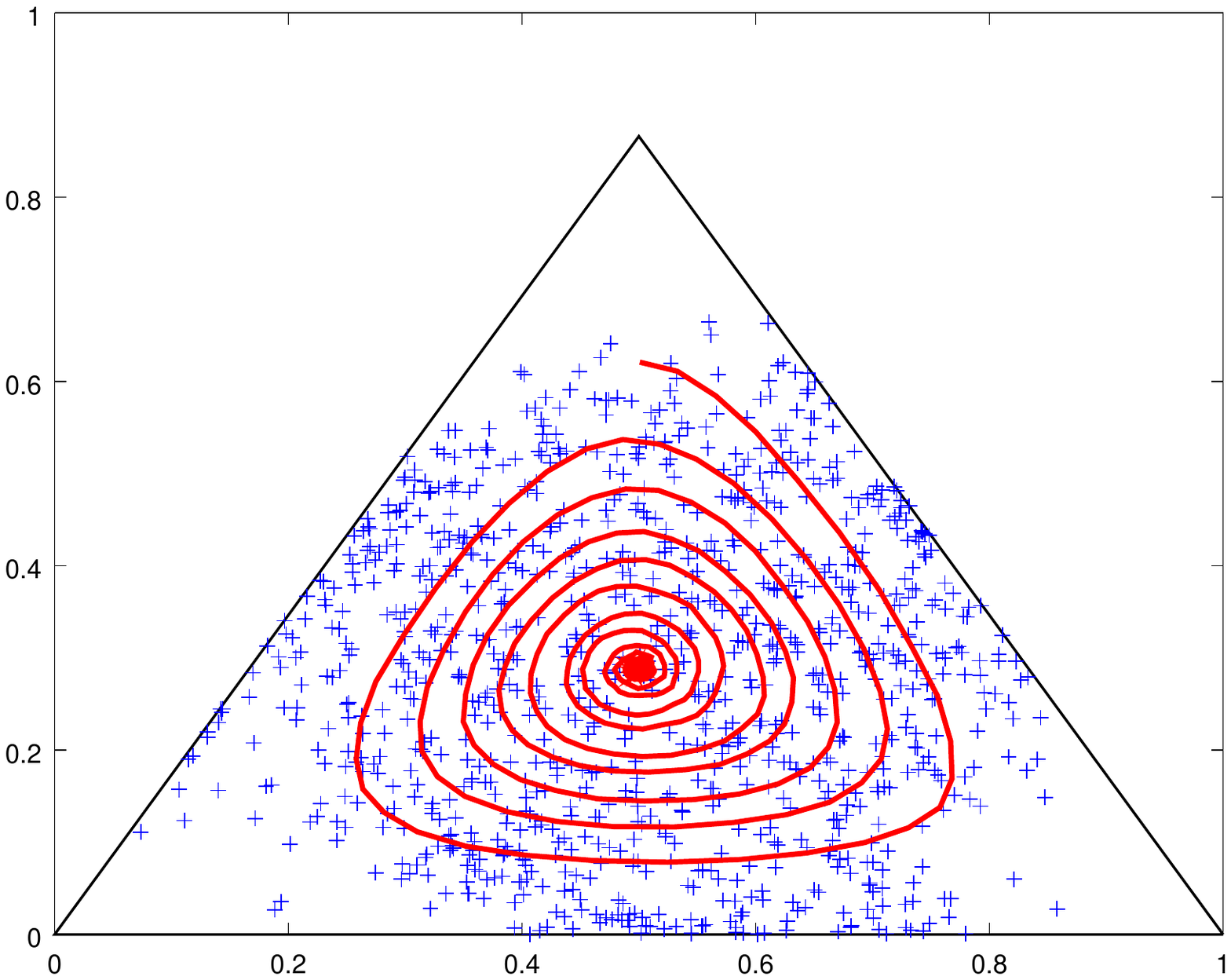}}
  \end{minipage}
  \caption{Diffusion of the players on the simplex of strategies. The first image is the initial distribution of the players, and
  the next ones are snapshots at steps 5, 10, 45, 100, and 1000. In each step, every  agent is sequentially selected
  and randomly  matched  with  another player.
   The solid line correspond to
  the trajectory of the mean strategy of the population. }
\end{center}
\end{figure}

\newpage

\section{Wealth distribution among the agents}

In this section we present the effect of the games on the wealth distribution curves. We start with  two strategies games and then we
consider
the rock-paper-scissors game.  Let us observe that in the former case, the players learn how to play the optimal pure strategy,
and then the wealth distribution frozen  since the payoff at equilibrium is zero; we are interested here in the initial advantage of a player who is close to the optimal strategy.
 In the latter
  case, the agents play a mixed strategy, and we compute the expected value of the variance
 due to the fluctuations of agents wealth; we find a remarkable agreement with the simulations.

Let us  introduce the family of probability  densities $\{P_t(w)\}_{t\ge 0}$,
and the probability to find an agent whose wealth belong to some interval
$(a,b)$ at time $t$  is given by
$$
\int_a^b P_t(w)dw.
$$

Now we  compute  the variance of $P_t(w)$ in order to compare it with our simulations. We have the following result:

\begin{thm}\label{varianza}
Let $G$ be the matrix of a symmetric zero sum game, and let us consider the wealth evolution process generated by the
  microscopic interaction rule \eqref{UpDatePayOff}.  Let us assume that the players are playing a mixed Nash equilibrium $p$,
  and that they are selected at random from the distribution $P_t$.
  Then, for sufficiently small $\ve$, we have
$$
\lim_{t\to \infty} \E(w^2)= \frac{\langle G^+\rangle}{ -\ve\langle G^{+ 2} \rangle +    \langle G^+\rangle},
$$
where
\begin{align*}
\langle G^\pm \rangle & = \ve^{-1} \iint pG_\ve^\pm p' \, P_t(w)P_t(w')dwdw' \\
 \langle G^{\pm 2} \rangle  & = \ve^{-2} \iint (pG_\ve^\pm p')^2 \, P_t(w)P_t(w')dwdw',
\end{align*}
\end{thm}

\begin{proof}
Using the interaction rule, and the distribution of players we have
\begin{align*}
  \frac{d}{dt} \E( w^2 ) = &\frac{d}{dt}\int w^2 P_t(w)dw \\
 =  &  \iint \big[ (w+pG_\ve^+p' w'-pG_\ve^-p' w)^2 -w^2 \big]  P_t(w)P_t(w')dwdw'
 \\
 =  &  \iint \big[ (pG_\ve^+p')^2 (w')^2  +2ww'(pG_\ve^+p') \big]  P_t(w)P_t(w')dwdw'
 \\
 & +   \iint \big[   (pG_\ve^-p')^2  -2(pG_\ve^-p') \big] (w)^2 P_t(w)P_t(w')dwdw' \\
   &  -  \iint   (pG_\ve^+p')(pG_\ve^-p')  w w' P_t(w)P_t(w')dwdw' .
\end{align*}

We use now that $\E(w^2)= \E(w^{' 2})$, $\E(w)=\E(w')=1$ and that $$ \iint(pG_\ve^+p')(pG_\ve^-p') w w' P_t(w)P_t(w')dwdw'=0,$$ since for each
pure strategy, one of the two games is equal to zero.
 Therefore,
\begin{align*}
  \frac{d}{dt} \E( w^2 )
 =  & \ve^2\langle G^{+ 2} \rangle   \E( w^2 )  +2 \ve\langle G^+\rangle
 \\
 & +  \ve^2\langle G^{- 2} \rangle   \E( w^2 )    -2 \ve\langle G^-\rangle   \E( w^2 )
   \end{align*}

We observe that $u(t) = \E(w^2)$ is a solution of the ordinary differential equation
$$
\frac{d}{dt}u(t) = A u(t) +B,
$$
for $A= \ve^2\langle G^{+ 2} \rangle +  \ve^2\langle G^{- 2} \rangle-2 \ve\langle G^-\rangle$, and $B=2 \ve\langle G^+\rangle$,
with some initial condition $u_0=\E(w_0^2)$ depending on the initial distribution of wealth.
The explicit solution is
$$
\E(w^2 )=C(A,u_0)exp(At) - A^{-1}B,
$$
and
$$
\lim_{t\to \infty} \E(w^2)= \frac{-2\langle G^+\rangle}{\ve\langle G^{+ 2} \rangle +  \ve\langle G^{- 2} \rangle-2 \langle G^-\rangle}
$$
since $A<0$ for $\ve$ small enough.

The result follows since $ \langle G^+\rangle= \langle G^-\rangle$ and $\langle G^{+ 2} \rangle=\langle G^{- 2} \rangle$
due to the symmetry of the game, and the proof is finished.
\end{proof}

Let us remark that the previous proof holds in the limit of infinitely many players. Of course, if the number of players is finite, we can understand
$P_t$ as a sum of Dirac's delta functions.

\subsection{Simulations for $2\times 2$ games}

\,

We have considered a symmetric game with two strategies $s_1$ and $s_2$, and the pay-off
matrix $G$   has the  form
$$ G = \begin{pmatrix}
0 & \ve \\
-\ve & 0
\end{pmatrix} $$
for $\ve\in (0,1)$.

As was proved in Subsection \ref{subdospordos}, we observe that  all the players learn how
  to play the optimal strategy. As a consequence, the  wealth distribution reaches an equilibrium which  is highly dependent on the
duration of  the transient state while the agents learn the optimal way of play, and then is frozen. So, it is interesting to study
the effect of the strategy update step $\delta$ on the steady state.

In the following figures we  have run Algorithm 1 with $10^4$ agents, each agent with wealth $1$, and $\ve=0.1$. We have fixed the vector of initial strategies for each
player and we have performed 100 simulations starting with these initial strategies.
In Figure 
\ref{finalwealth2por2} we observe the mean  steady state curve for these simulations, for $\delta=0.1$ (left),
$\delta=0.01$ (center), and $\delta=0.001$ (right).

\begin{figure}
 \centering
   \includegraphics[angle=0,width=.32\textwidth]{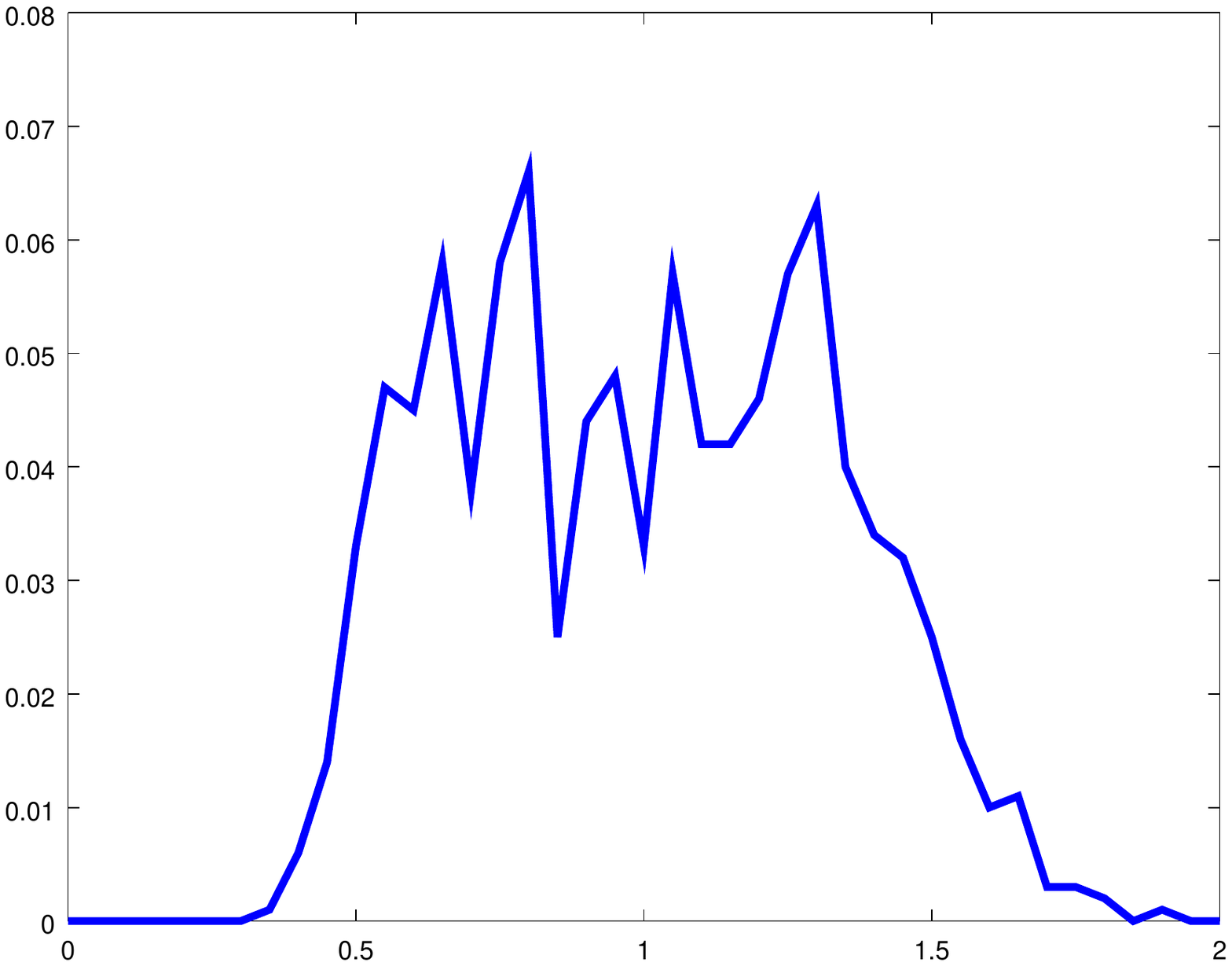}
   \includegraphics[angle=0,width=.32\textwidth]{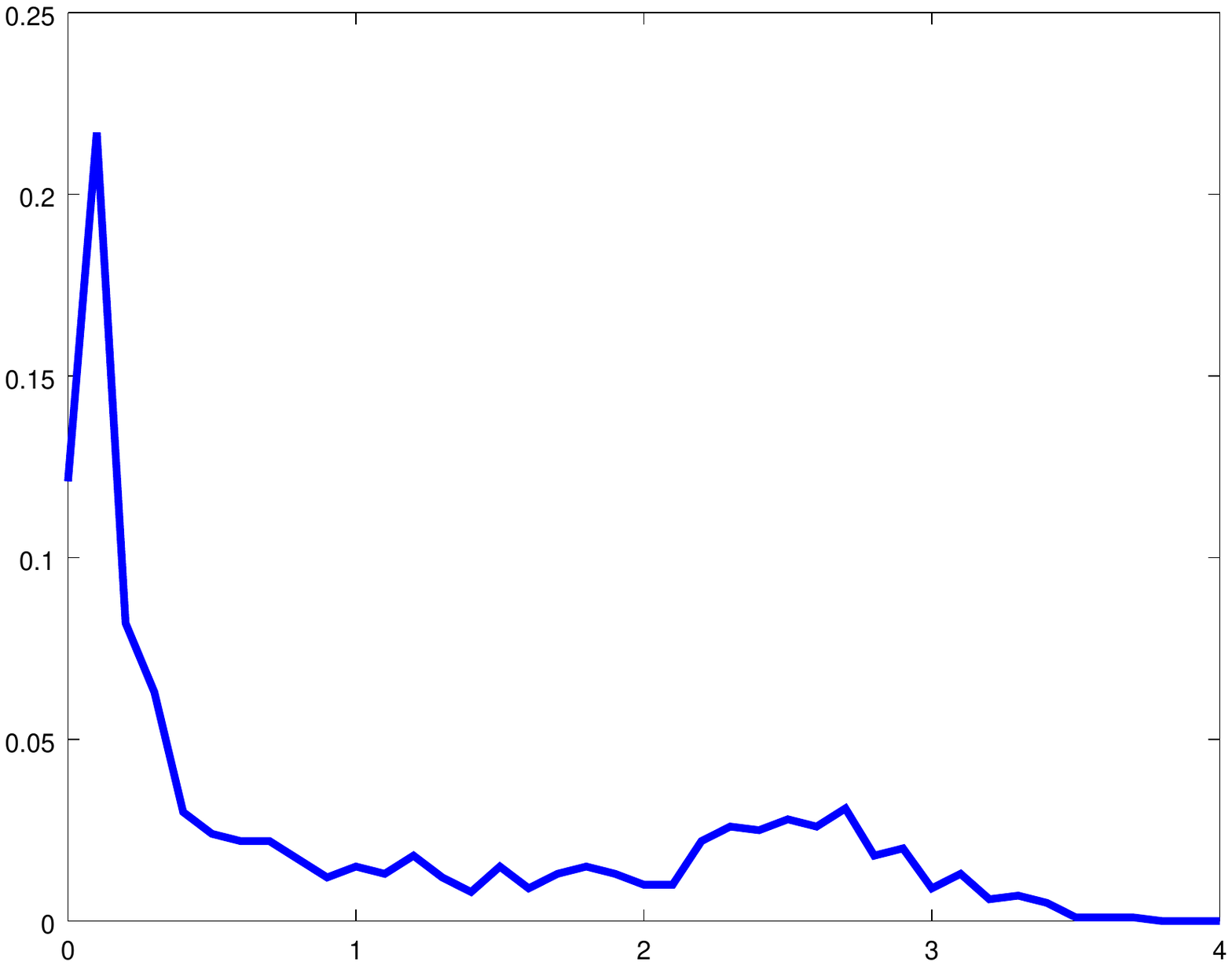}
   \includegraphics[angle=0,width=.32\textwidth]{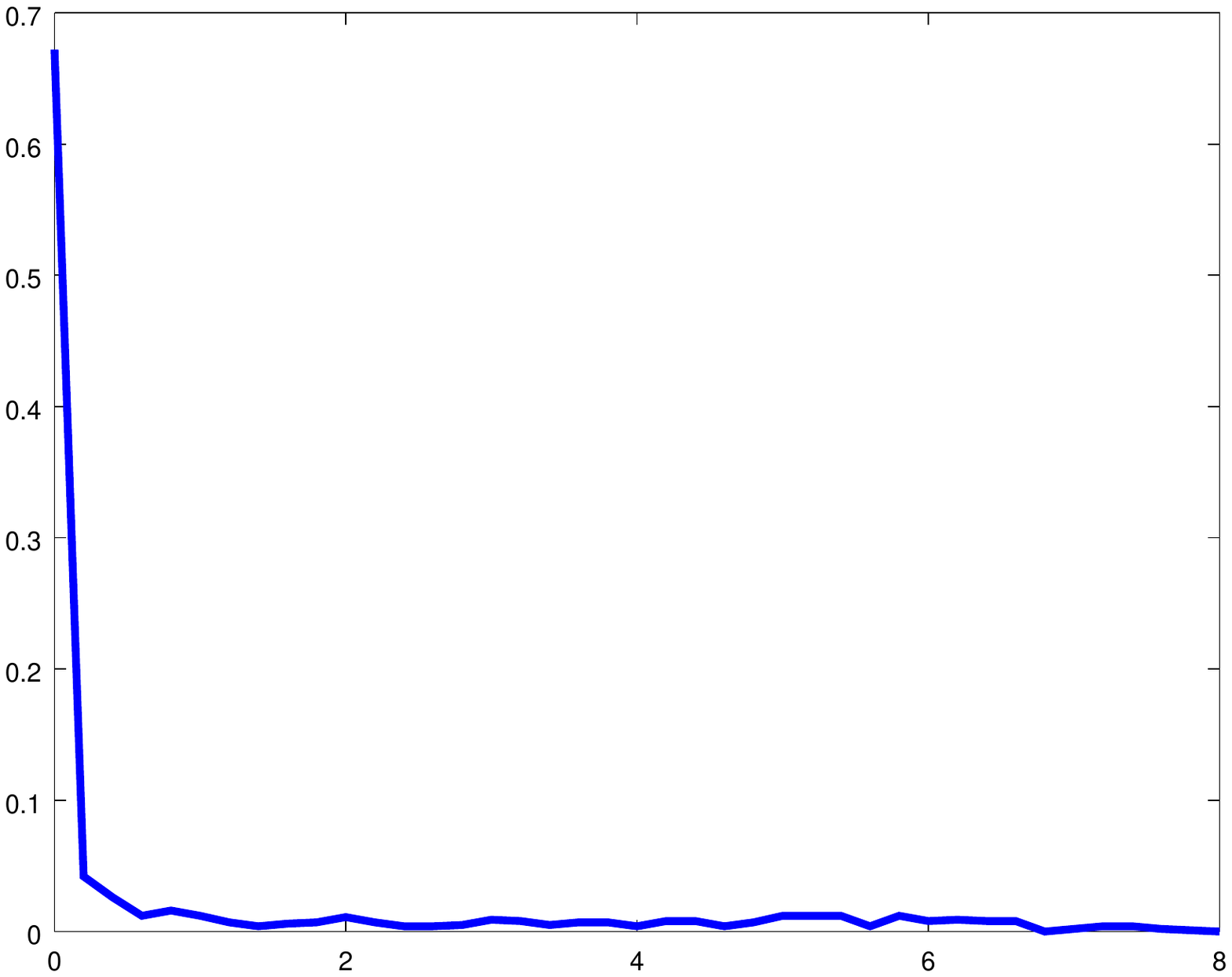}
 \caption{Steady state of the wealth for different values of $\delta$, $\delta=0.1$ (left),
$\delta=0.01$ (center), and $\delta=0.001$ (right). We averaged 100 simulations with $N=10^4$ agents and the same initial distribution
of strategies.}\label{finalwealth2por2}
\end{figure}

In Figure \ref{wealstratinitial2por2} we show the mean wealth obtained at the steady state
for the agents as a function of their initial strategy. The horizontal axis represent the probability
that an agent play the optimal strategy, and clearly the ones which started with a high probability to play it become more wealthier.
Here, the value of $\delta$ determines the duration of the transient state, where the wealth exchange takes place. We have $\delta=0.1$ (left),
$\delta=0.01$ (center), and $\delta=0.001$ (right).

\begin{figure}
 \centering
   \includegraphics[angle=0,width=.32\textwidth]{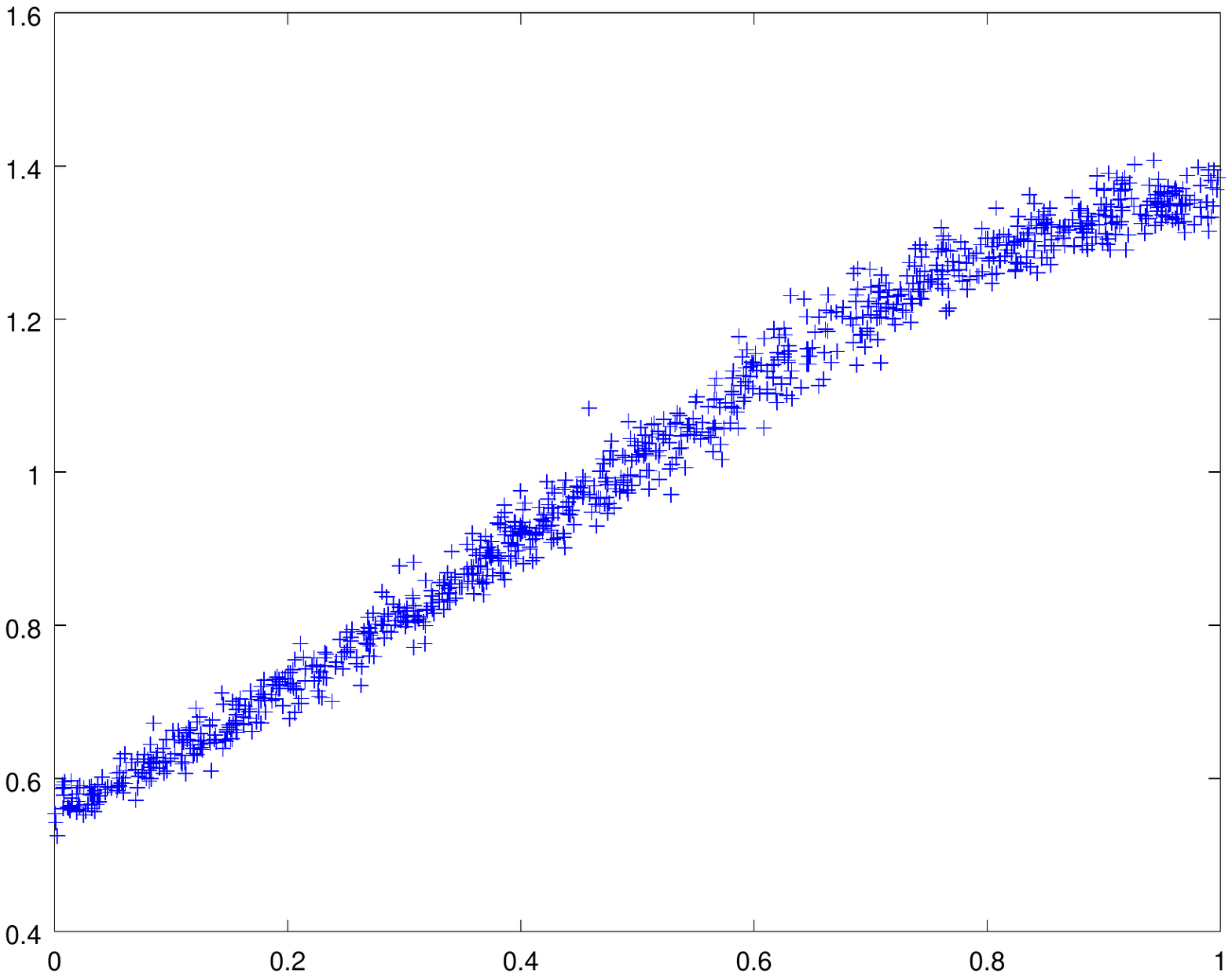}
   \includegraphics[angle=0,width=.32\textwidth]{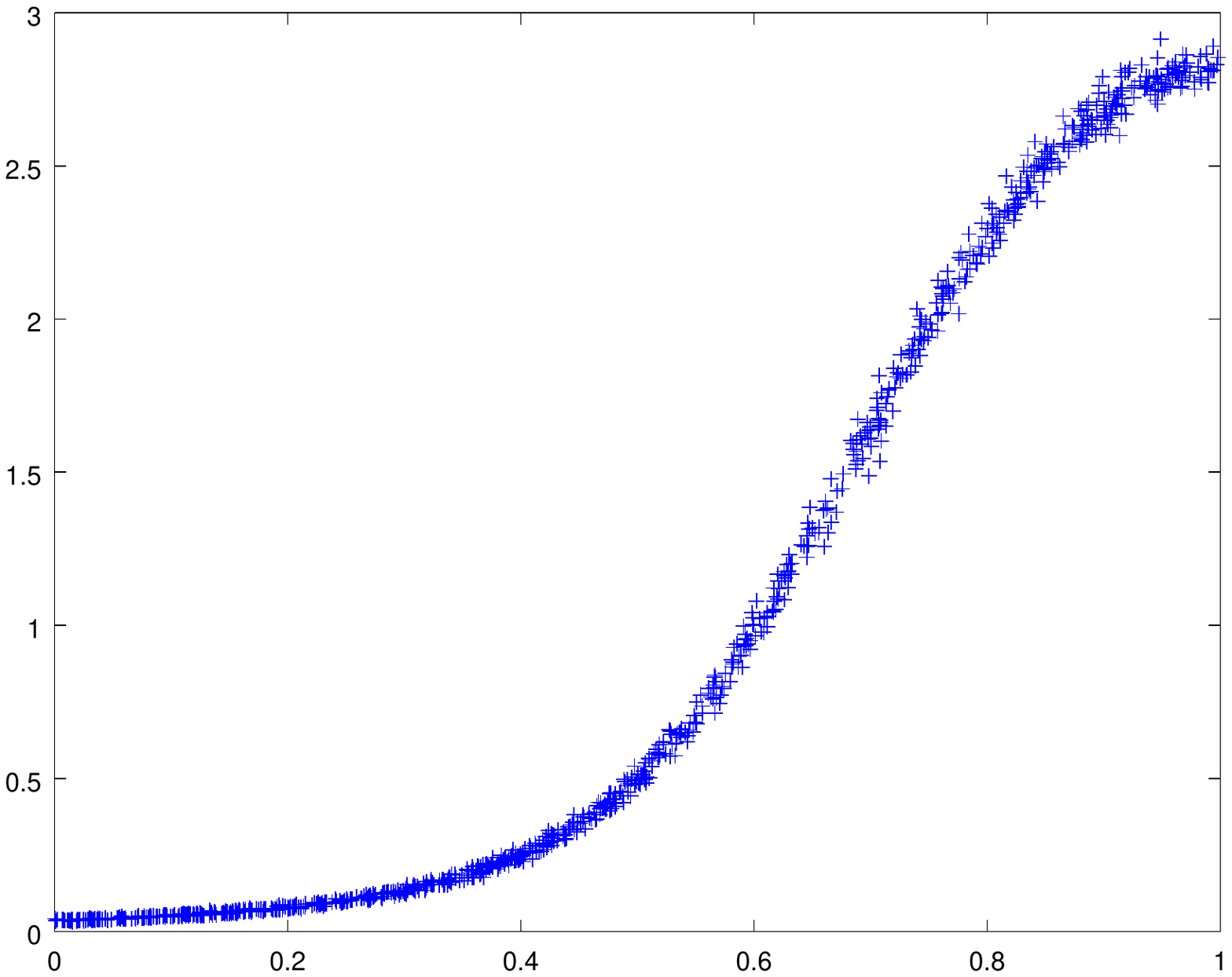}
   \includegraphics[angle=0,width=.32\textwidth]{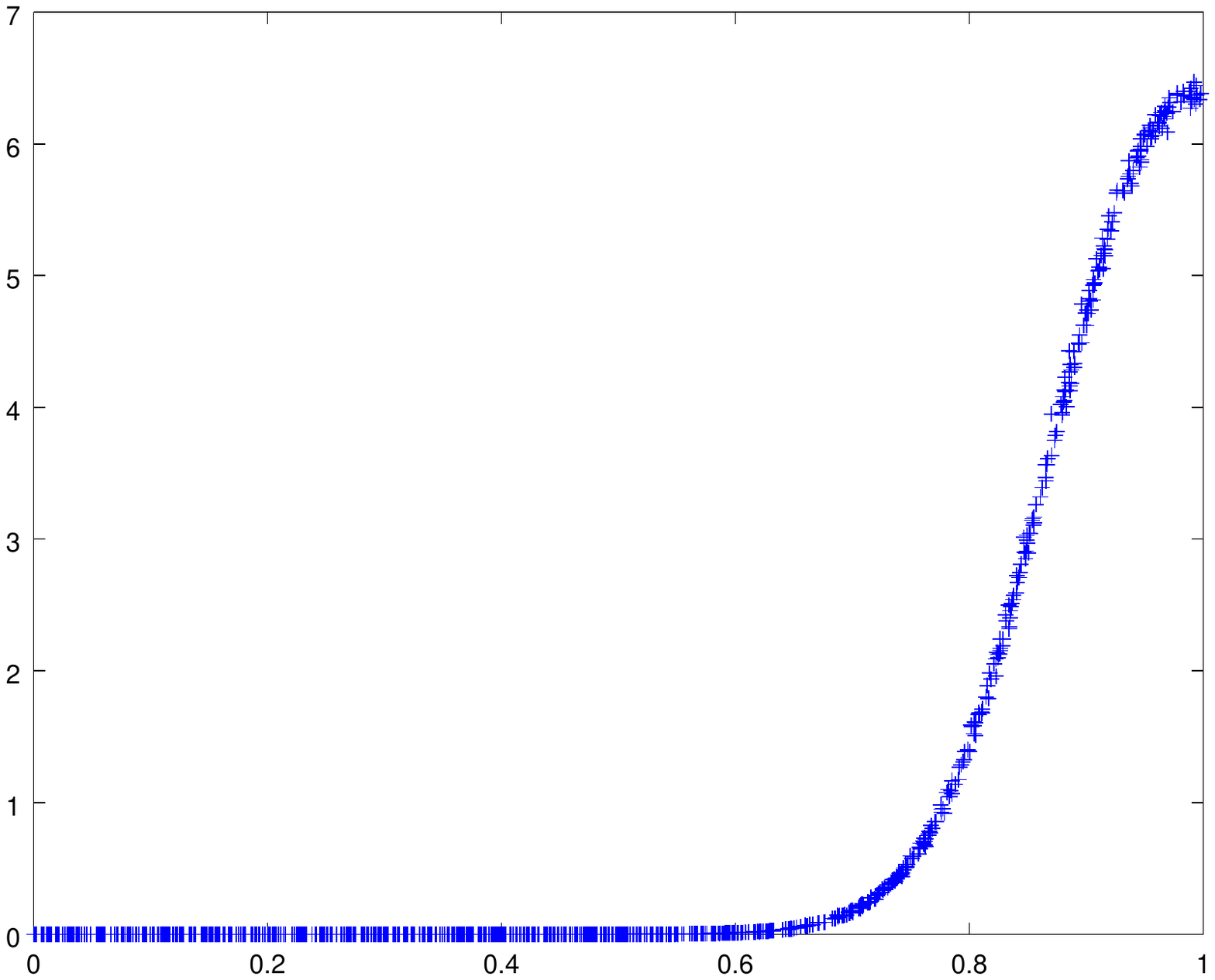}
 \caption{Mean wealth of agents (vertical axis) depending on their initial probability to play the optimal strategy (horizontal axis), for different
 values of $\delta$, $\delta=0.1$ (left),
$\delta=0.01$ (center), and $\delta=0.001$ (right). We averaged 100 simulations with $N=10^4$ agents and the same initial distribution
of strategies.}\label{wealstratinitial2por2} 
\end{figure}

\subsection{The Rock-paper-scissors game}

\,

Let us consider now the Rock-paper-scissors game from Subsection \ref{rps}. We fix  $a=b=1$, and we scale the game with different parameters
$\ve\in (0,1)$.

 In Figure \ref{g1} we show several histograms, for different values of $\ve$, of the wealth distribution for $N=10^4$ agents, averaged over $10$ realizations
starting all of them with the same initial strategies, and each agent has initial wealth equal to $1$. Let us note that we have less equalitarian
distributions when $\ve$ approaches $1$. In fact, for $\ve=1$, all the wealth is transferred to the winner player, and in the long run, it accumulates in
a single player.

The solid lines in Figure \ref{g1} represent the Gamma function
$$
 \Gamma(x) = \frac{v^v\exp(-vx)}{\Gamma(v)x^{1-v}},$$
where $v=v(\ve)$ is the  sample variance of the data in each case. We show in Table \ref{varioseps} that the sample variance and
the theoretical value $\ve (1-\ve)^{-1}$ are very close.

The value $v(\ve)=\ve (1-\ve)^{-1}$ can be computed from Theorem \ref{varianza}, or using directly  the interaction rule, from
$$
w_i^* = w_i +  \ve(g_{lm}^+w_j - g_{lm}^-w_i).
$$

 Let us note that the expected value of
$w_i$ is $\E(w_i)=1$, hence the variance is
$$Var(w_i)= \E(w_i^2)-1.$$

Let us call $V=\E(w_i^2)$. The linearity of the expected value, and the independence implies
\begin{align*}
 \E([w_i^*]^2)
= & \E([w_i +  \ve g_{lm}^+w_j - \ve g_{lm}^-w_i]^2) \\
 = & \E([w_i]^2) + \E([\ve g_{lm}^+w_j]^2) +\E([\ve g_{lm}^-w_i]^2) +
 2 \E(\ve g_{lm}^+w_j w_i)  \\ 
 & \quad - 2\E(\ve  g_{lm}^-w_i^2)
 -2 \E( \ve^2 g_{lm}^+ g_{lm}^-w_j w_i) .\end{align*}
Therefore,
$$
 V=  V + \ve^2 \E([g_{lm}^+]^2)V +\ve^2 \E([g_{lm}^-]^2)V +
\ve 2 \E(g_{lm}^+)  - 2\ve \E( g_{lm}^-)V  -\ve^2 2 \E(  g_{lm}^+ g_{lm}^-),
$$
which gives
$$ V =\frac{  2 \E(g_{lm}^+) - \ve 2 \E( g_{lm}^+ g_{lm}^-)  }{2 \E( g_{lm}^-)
 - \ve \E([g_{lm}^+]^2)
-\ve \E([g_{lm}^-]^2)   }
$$

Using that
$$ \E(g_{lm}^+)=\E(g_{lm}^-)=  \E([g_{lm}^-]^2) =  \E([g_{lm}^+]^2)=\frac{1}{3},$$
and  $\E( g_{lm}^+ g_{lm}^-)= 0$ since one of them must be zero, we get
$$
 V =\frac{  \frac{2}{3}   }{ \frac{2}{3}
 - \frac{2\ve}{3}   } = \frac{1}{1-\ve}
$$
and thus
$$
Var(w_i) = \frac{\ve}{1-\ve},
$$
as stated before.

\begin{figure}
\begin{center}
  \begin{minipage}{0.45\textwidth}
    {\includegraphics[width=\textwidth]{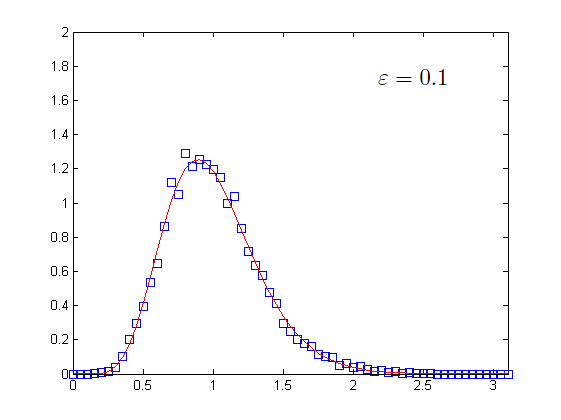}}
    {\includegraphics[width=\textwidth]{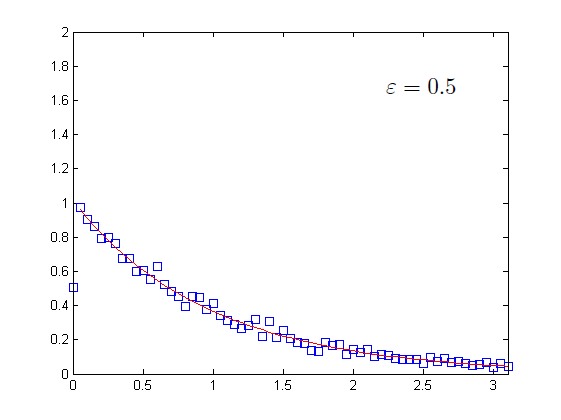}}
  \end{minipage}
  \hfill
  \begin{minipage}{0.45\textwidth}
    {\includegraphics[width=\textwidth]{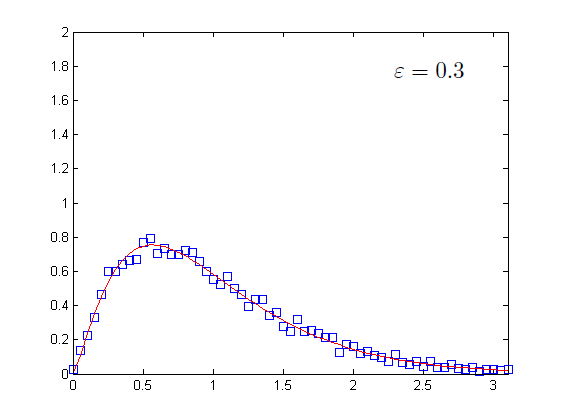}}
    {\includegraphics[width=\textwidth]{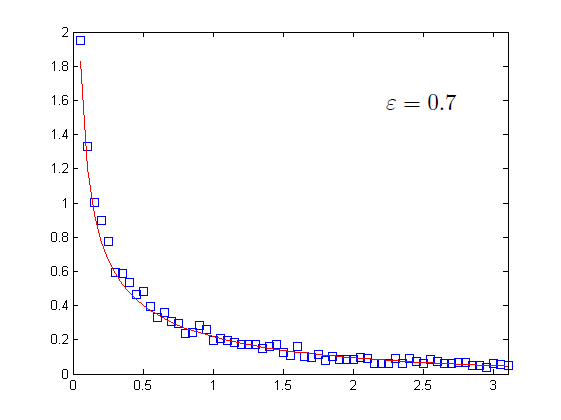}}
  \end{minipage}
  \caption{Histogram of the wealth distribution with $N=10^4$ agents, averaged over $10$ realizations, starting each one with wealth equal to $1$.
The solid line represent the Gamma function of parameter $v=v(\ve)$.}\label{g1} 
\end{center}
\end{figure}

  \begin{table}
  \begin{center}
    \begin{tabular}{|c|c|c|c|}
     \hline
       Value of $\ve$ & Sample variance  & Theoretical value & Relative error \\ \hline
       0.1 & 0.1114 & 0.1111 & 0.0030 \\ \hline
       0.3 & 0.4365 & 0.4286 & 0.0185 \\ \hline
       0.5 & 1.0098 & 1.0000 & 0.0098 \\ \hline
       0.7 & 2.3635 & 2.3333 & 0.0129 \\ \hline
       0.9 & 9.0967 & 9.0000 & 0.0107  \\ \hline
     \end{tabular}\vspace{0.3cm}
     \caption{Variance of Gamma distribution depending of the value of $\ve$.}\label{varioseps} 
  \end{center}
\end{table}

\section{Conclusions}

In this work we have studied the wealth distribution problem when the agents interact
through a zero sum game, which determine the wealth transfer. Moreover, we have studied
the evolutive process, for agents using mixed strategies.

 We have proposed a Pavlovian type of agents, which update their strategies myopically as a result of the
 last game they played, without trying to find the best response, or changing strategies proportionally to looses and gains,
 but reinforcing the winning strategy and penalizing the loosing strategy. A theoretical analysis of this simple dynamic gives a
 system of ordinary differential equations for the evolution of each agent strategies, and the mean strategy
 of the population solves a system very close to the
 replicator equation.

The wealth evolution of the population is similar to the one observed in more simple models
where the interaction between agents is restricted to a random exchange of wealth. By
scaling the game payoff we reproduce the saving propensity phenomena, and different
Gamma-like curves appears as the steady state as a function of the scaling parameter.

Two different scenarios appear depending on the Nash equilibrium of the game. Since the
game is symmetric, its value is zero, and when the Nash equilibrium is a pure strategy, the
wealth transfer stops whenever the agents learn how to play in the optimal way.  In this case
we have more equalitarian distribution of wealth when they quickly update their strategies,
and less equalitarian distribution when the evolutionary process is slow.

However, when the Nash equilibrium is a mixed strategy, the wealth  transfer continues due
to the random results of each individual game, and the steady state for the wealth
distribution resembles the curves obtained using purely random ways to determine the
wealth transfer. The shape of the steady state varies with the proportion of the agent's
wealth that is involved in each transaction, and more equalitarian distributions are obtained
when this proportion is small, and if we allow to exchange almost all the wealth in a single
game, we get a distribution strongly concentrated near zero, since wealth is accumulated in
few individuals.

\bigskip

Finally, two directions seem promising. In the limit of infinitely many players we can expect a
system of two coupled Boltzmann equations, one describing the  evolutionary process on the
simplex of strategies, and the other one modeling the wealth exchange.

On the other hand, new situations can occur whenever the agents play an arbitrary game.
Beside the dichotomy between pure and mixed equilibria, we can have more than one
equilibria, with different payoffs, and different kind of wealth distribution can appear
depending on the limit equilibria selected by the dynamics. Let us remak that in this case the
wealth is not conserved, and some kind of rescaling will be needed in order to study the
asymptotic behavior of the wealth distribution.

\section*{Acknowledgments}

This paper was partially supported by grants UBACyT 20020130100283BA,
CONICET PIP 11220150100032CO and ANPCyT PICT 2012-0153.
J. P. Pinasco and N. Saintier members of CONICET. M. Rodriguez Cartabia is a  Fellow of CONICET.

\bigskip
\hfill\break
 Juan Pablo Pinasco \hfill\break
Departamento de Matem\'atica, FCEyN - Universidad de Buenos Aires \hfill\break and  IMAS,
CONICET-UBA \hfill\break    Ciudad Universitaria, Pabell\'on I (1428) Av. Cantilo s/n.
\hfill\break   Buenos Aires, Argentina. \hfill\break e-mail:{\tt jpinasco@dm.uba.ar}

\bigskip
\hfill\break
 Mauro Rodriguez Cartabia \hfill\break
Departamento de Matem\'atica, FCEyN - Universidad de Buenos Aires  \hfill\break
 and  IMAS,  CONICET-UBA  \hfill\break  Ciudad Universitaria, Pabell\'on I (1428) Av.
Cantilo s/n. \hfill\break  Buenos Aires, Argentina.\hfill\break e-mail:{\tt
mrodriguezcartabia@gmail.com}

\bigskip
\hfill\break
 Nicolas Saintier  \hfill\break
Departamento de Matem\'atica, FCEyN - Universidad de Buenos Aires   \hfill\break and IMAS,
CONICET-UBA  \hfill\break   Ciudad Universitaria, Pabell\'on I (1428) Av. Cantilo s/n.
\hfill\break   Buenos Aires, Argentina.\hfill\break e-mail:{\tt nsaintie@dm.uba.ar}

\end{document}